\documentclass[conference,a4paper]{IEEEtran}

\usepackage{graphicx}
\usepackage{cite}
\usepackage{multirow}
\usepackage{epstopdf}

\usepackage{amsthm}
\usepackage{amsmath}
\usepackage{amsfonts}
\usepackage{amssymb}
\usepackage{algorithm}
\usepackage{algorithmic}
\usepackage{verbatim}
\usepackage{multirow}
\usepackage{framed}
\usepackage{balance}

\usepackage{graphicx}
\usepackage{epstopdf}
\graphicspath{{fig/}}
\usepackage{calc}
\usepackage{color}
\usepackage[caption=false,font=footnotesize]{subfig}

\newtheorem{theorem}{Theorem}

\newtheorem{remark}{Remark}

\newtheorem{corollary}{Corollary}
\newtheorem{example}{Example}

\newtheorem{proposition}{Proposition}
\newtheorem{stopping criterion}{stopping criterion}

\newcommand{\beq}{\begin{equation}}
\newcommand{\eeq}{\end{equation}}
\newcommand{\beqnn}{\begin{equation*}}
\newcommand{\eeqnn}{\end{equation*}}
\newcommand{\beqy}{\begin{eqnarray}}
\newcommand{\eeqy}{\end{eqnarray}}
\newcommand{\beqynn}{\begin{eqnarray*}}
\newcommand{\eeqynn}{\end{eqnarray*}}
\newcommand{\bit}{\begin{itemize}}
\newcommand{\eit}{\end{itemize}}
\newcommand{\ben}{\begin{enumerate}}
\newcommand{\een}{\end{enumerate}}
\newcommand{\bex}{\begin{example}}
\newcommand{\eex}{\end{example}}

\newcommand{\balg}[1]{\begin{algorithm} \caption{#1}}
\newcommand{\ealg}{\end{algorithm}}

\newcommand{\balgc}{\begin{algorithmic}[1]}
\newcommand{\ealgc}{\end{algorithmic}}

\newcommand{\bary}{\begin{array}}
\newcommand{\eary}{\end{array}}
\newcommand{\bmx}{\begin{bmatrix}}
\newcommand{\emx}{\end{bmatrix}}
\newcommand{\bsmx}{\left[\begin{smallmatrix}}
\newcommand{\esmx}{\end{smallmatrix}\right]}
\newcommand{\bmxc}[1]{\left[\begin{array}{@{}#1@{}}}
\newcommand{\emxc}{\end{array}\right]}
\newcommand{\bcn}{\begin{center}}
\newcommand{\ecn}{\end{center}}

\newcommand{\diag}{\mathrm{diag}}

\newcommand{\bigO}{{\mathcal{O}}}

\newcommand{\A}{\boldsymbol{A}}
\newcommand{\B}{\boldsymbol{B}}

\newcommand{\F}{\boldsymbol{F}}
\newcommand{\G}{\boldsymbol{G}}
\renewcommand{\H}{\boldsymbol{H}}
\newcommand{\I}{\boldsymbol{I}}

\newcommand{\R}{\boldsymbol{R}}

\newcommand{\X}{{\boldsymbol{X}}}

\newcommand{\Z}{\boldsymbol{Z}}
\renewcommand{\a}{\boldsymbol{a}}

\newcommand{\x}{{\boldsymbol{x}}}

\newcommand{\0}{{\boldsymbol{0}}}

\newcommand{\bbR}{{\bar{\R}}}

 \usepackage{cite,color}
\DeclareMathOperator*{\argmin}{arg\,min}
\DeclareMathOperator*{\argmax}{arg\,max}

\begin{document}

\title{Some Properties of Successive Minima and Their Applications}


%
%

\author{\IEEEauthorblockN{Jinming~Wen}
\IEEEauthorblockA{
Edward S. Rogers Sr. Department of Electrical and Computer Engineering\\
University of Toronto, Toronto, M5S 3G4, Canada \\
(E-mail: jinming.wen@utoronto.ca)}


}

\maketitle

\begin{abstract}
A lattice is a set of all the integer linear combinations of certain linearly independent vectors.
One of the most important concepts on lattice is the successive minima which is of vital importance
from both theoretical and practical applications points of view.
In this paper, we first study some properties of successive minima and then employ some of them to improve the suboptimal algorithm for solving an
optimization problem about maximizing the achievable rate of the integer-forcing strategy for
cloud radio access networks in \cite{BakN17}.
\end{abstract}


\IEEEpeerreviewmaketitle
\section{Introduction}
\label{sec:Introduction}


A lattice is a set of all the integer linear combinations of certain linearly independent vectors.
Specifically, for any full column rank matrix $\A\in \mathbb{R}^{m\times n}$ ($m\geq n$),
the lattice $\mathcal{L}(\A)$ generated by $\A$ is defined by
\beq
\label{e:latticeA}
\mathcal{L}(\A)=\{\A\x|\x \in \mathbb{Z}^n\},
\eeq
and $\A$ is called as the basis matrix of $\mathcal{L}(\A)$,
whose dimension is defined as the rank of $\A$.


One of the most important concepts on lattice is the successive minima.
Specifically, for any $n$-dimensional $\mathcal{L}(\A)$, its $i$-th ($1\leq i\leq n$)
successive minimum $\lambda_i(\A)$ is defined as the smallest $r$ such that the closed
$n$-dimensional ball $\mathbb{B}(\0,r)$ of radius $r$
centered at the origin contains $i$ linearly independent lattice vectors.


Finding a vector whose length equals to a certain successive minimum is needed in a variety of applications. For example, in communications (see, e.g., \cite{Mow94}) and cryptography
(see, e.g., \cite{DanG02}),
one frequently  needs to solve the following shortest vector problem (SVP) on $\mathcal{L}(\A)$:
\beqnn
\min_{\x \in\mathbb{Z}^n\backslash \{\0\}} \|\A\x\|_2,
\eeqnn
whose solution $\x$ satisfies $\|\A\x\|_2=\lambda_1(\A)$.
In some other applications, such as,
integer-forcing (IF) linear receiver design \cite{ZhaNEG14},
(after some transformations) one needs to solve a
Shortest Independent Vector Problem (SIVP) on lattice $\mathcal{L}(\A)$, i.e.,
finding an invertible matrix
$\X=[\x_1,\ldots, \x_{n}]\in \mathbb{Z}^{n\times n}$ such that
\beq
\label{e:SIVP}
\max_{1\leq i\leq n} \|\A\x_i\|_2\leq \lambda_{n}(\A).
\eeq
A closely related problem to the SIVP is a Successive Minima Probem (SMP), i.e.,
finding an invertible matrix
$\X=[\x_1,\ldots, \x_{n}]\in \mathbb{Z}^{n\times n}$ such that
\beq
\label{e:SMP}
\|\A\x_i\|_2= \lambda_{i}(\A),\quad i=1,2,\ldots, n.
\eeq
Solving an SMP is needed in some practical applications, such as physical-layer network coding \cite{FenSK13},
the expanded compute-and-forward framework \cite{NazCNC16} and IF source coding \cite{OrdE17}.


Cloud radio access networks (C-RANs) is a promising framework for 5G wireless communication systems.
Recently, an IF framework with two architectures for uplink C-RANs has been proposed in \cite{BakN17}.
Simulations in \cite{BakN17} indicate that for the scenario where channel state is available
to the receivers only, the two architectures can nearly match and often outperforms
Wyner-Ziv-based strategies, respectively.

Successive minima is of vital importance from both theoretical and application points of view.
Thus, this paper aims to develop some properties of successive minima.
Some of them are useful for IF design for uplink C-RANs.

The rest of the paper is organized as follows.
We develop some properties of successive minima in Section \ref{s:property},
and use  some of them to improve the suboptimal algorithm for IF design for C-RAN in \cite{BakN17}
in Section \ref{s:IF}.
Finally, conclusions are given in Section \ref{s:con}.

{\it Notation.}
Let $\mathbb{R}^{m\times n}$ and $\mathbb{Z}^{m\times n}$ respectively
stand for the space of the $m\times n$ real and integer matrices.
Let $\mathbb{R}^n$ and $\mathbb{Z}^n$ denote the space of the $n$-dimensional  real
and integer column vectors, respectively.
For a symmetric positive definite (SPD) matrix $\G\in\mathbb{R}^{n\times n}$, we use
chol($\G$) to denote the Cholesky factor of $\G$.
For a matrix $\A$, let $a_{ij}$ denote its element at row  $i$ and column  $j$,
$\a_i$ be its $i$-th column.
For a vector $\x$, let $x_i$ be its $i$-th element.

\section{Some Properties of Successive Minimum}
\label{s:property}

In this section, we first develop the monotonic property of successive minima.
Then, we propose a lower and an upper bound on them.
Some of these properties will be used in the next section for IF design for C-RAN.

To prove our proposed properties of successive minima, we need to introduce the following well-known property of successive minima:
\beq
\label{e:nondecrease}
\lambda_1(\A)\leq \lambda_2(\A)\leq \cdots\leq \lambda_n(\A)
\eeq
for any full column rank matrix $\A\in \mathbb{R}^{m\times n}$.
In fact, \eqref{e:nondecrease} can be easily seen from the definition of successive minima.


\subsection{Monotonic property of successive minima}
In the following, we develop some properties of successive minima.
Since in communications, it is often that one needs to solve an SMP on a lattice whose basis matrix
is not explicitly given, but is the Cholesky factor of an SPD matrix, in the sequel,
we develop some properties of successive minima of some lattices whose bases matrices are
the Cholesky factors of some SPD matrices.


\begin{theorem}
\label{t:adlb}
Suppose that $\G_1, \G_2\in \mathbb{R}^{n\times n}$ are SPD matrices.
Denote
\[
\R_1=\mbox{chol}(\G_1), \R_2=\mbox{chol}(\G_2), \R_3=\mbox{chol}(\G_1+\G_2).
\]
Then, for $i=1,2,\ldots, n$, we have
\begin{align}
\label{e:adlb}
\lambda_i(\R_3)\geq \max\{\sqrt{\lambda_i^2(\R_1)+\lambda_1^2(\R_2)}, \sqrt{\lambda_i^2(\R_2)+\lambda_1^2(\R_1)}\}.
\end{align}
\end{theorem}

\begin{proof}
Since both $\G_1$ and $\G_2$ are SPD matrices, $\G_1+\G_2$ is also an SPD matrix.
Hence, $\R_3=\mbox{chol}(\G_1+\G_2)$ is an invertible matrix,
implying that $\mathcal{L}(\R_3)$ is well-defined.

To show \eqref{e:adlb}, it is equivalent to show
\beq
\label{e:adlb11}
\lambda_i(\R_3)\geq \sqrt{\lambda_i^2(\R_1)+\lambda_1^2(\R_2)}
\eeq
and
\beq
\label{e:adlb12}
\lambda_i(\R_3)\geq \sqrt{\lambda_i^2(\R_2)+\lambda_1^2(\R_1)}.
\eeq
In the following, we only prove \eqref{e:adlb11} since \eqref{e:adlb12}
can be similarly proved.

Since the solution of the SMP on $\mathcal{L}(\R_3)$ (see \eqref{e:SMP}) always exists \cite{Cas12},
there exists an invertible matrix $\X=[\x_1,\ldots, \x_{n}]\in \mathbb{Z}^{n\times n}$ such that
\[
\|\R_3\x_i\|_2= \lambda_{i}(\R_3),\quad i=1,2,\ldots, n,
\]
which combing with $\R_3=$\mbox{chol}($\G_1+\G_2$) implies that
\[
\x_i^\top(\G_1+\G_2)\x_i= \lambda_{i}^2(\R_3),\quad i=1,2,\ldots, n.
\]

Let
\beq
\label{e:j}
j=\argmax_{1\leq k\leq i}\x_k^\top\G_1\x_k.
\eeq
Then, according to \eqref{e:nondecrease}, we have
\[
\lambda_{i}^2(\R_3)\geq \lambda_{j}^2(\R_3)=\x_j^\top\G_1\x_j+\x_j^\top\G_2\x_j.
\]
Since $\X$ is invertible, $\x_1,\cdots, \x_{i}$ are linearly independent,
so are $\R\x_1,\cdots, \R\x_{i}$.
Then, by \eqref{e:j} and the definition of successive minima, one can see that
\[
\x_j^\top\G_1\x_j=\|\R_1\x_j\|^2_2\geq \lambda_i^2(\R_1).
\]
Since $\G_2$ is an SPD matrix, $\R_2$ is invertible.
Then by the definition of successive minima and the fact that $\x_j$ is a nonzero integer vector, we can see
that
\[
\x_j^\top\G_2\x_j=\|\R_2\x_j\|^2_2\geq \lambda_{1}^2(\R_2).
\]
Hence, \eqref{e:adlb11} follows from the above three equations.
\end{proof}

\begin{remark}
\label{r:semid}
Note that it is not necessary that both $\G_1$ and $\G_2$ are SPD matrices in Theorem \ref{t:adlb}.
In fact, one of them is an SPD matrix and the other one is a non-invertible symmetric positive
semidefinite matrix is enough.
More specifically, if $\G_1$ is an SPD matrix and $\G_2$ is a non-invertible symmetric positive
semidefinite matrix, then \eqref{e:adlb} is reduced to $\lambda_i(\R_3)\geq\lambda_i(\R_1)$.
Similarly, if $\G_2$ is an SPD matrix and $\G_1$ is a non-invertible symmetric positive
semidefinite matrix, then \eqref{e:adlb} is reduced to $\lambda_i(\R_3)\geq\lambda_i(\R_2)$.
\end{remark}

We would like to point out that the equality in \eqref{e:adlb} is achievable.
For more details, see the following example.
\begin{example}
Suppose that $\alpha_1,\alpha_2,\cdots, \alpha_n$ satisfy
\beq
\label{e:alpha1n}
0<\alpha_1\leq \alpha_2\leq\cdots\leq\alpha_n.
\eeq
Let $\G_1$ be a diagonal matrix with its diagonal entries being $\alpha_1,\alpha_2,\cdots, \alpha_n$,
i.e., $\G_1=\diag(\alpha_1,\alpha_2,\cdots, \alpha_n)$, $\G_2=\beta\I$ for some $\beta>0$.
Then
\begin{align*}
\R_1&=\diag(\sqrt{\alpha_1},\sqrt{\alpha_2},\cdots, \sqrt{\alpha_n}), \,\;\R_2=\sqrt{\beta}\I,\\
\R_3&=\diag(\sqrt{\alpha_1+\beta},\sqrt{\alpha_2+\beta},\cdots, \sqrt{\alpha_n+\beta})
\end{align*}
which combing with \eqref{e:alpha1n} implies that
\[
\lambda_i(\R_1)=\sqrt{\alpha_i}, \,\;\lambda_i(\R_2)=\sqrt{\beta}, \,\; \lambda_i(\R_3)=\sqrt{\alpha_i+\beta}
\]
for $i=1,2,\ldots, n$. Hence,
\begin{align*}
\sqrt{\lambda_i^2(\R_1)+\lambda_1^2(\R_2)}&=\sqrt{\alpha_i+\beta}\\
&\geq \sqrt{\alpha_1+\beta}=\sqrt{\lambda_i^2(\R_2)+\lambda_1^2(\R_1)}.
\end{align*}
Thus, the equality in \eqref{e:adlb} is reached.
\end{example}

By using Theorem \ref{t:adlb}, we can prove the following theorem which  provides an upper
bound on the successive minima of a lattice whose basis matrix is given by the Cholesky factor of
the inverse of the sum of two SPD matrices.

\begin{theorem}
\label{t:adivlb}
Suppose that $\G_1, \G_2\in \mathbb{R}^{n\times n}$ are SPD matrices.
Denote $\hat{\R}_1=$chol($\G_1^{-1})$, $\hat{\R}_2=\mbox{chol}(\G_2^{-1})$,
$\hat{\R}_3=\mbox{chol}((\G_1+\G_2)^{-1})$ and
\begin{align*}
\hat{\R}_4&=\mbox{chol}(\G_1^{-1}(\G_1^{-1}+\G_2^{-1})^{-1}\G_1^{-1}),\\
\hat{\R}_5&=\mbox{chol}(\G_2^{-1}(\G_1^{-1}+\G_2^{-1})^{-1}\G_2^{-1}).
\end{align*}
Then, for $i=1,2,\ldots, n$, we have
\begin{align}
\label{e:adivlb1}
\lambda_i(\hat{\R}_1)&\geq \max\{\sqrt{\lambda_i^2(\hat{\R}_3)+\lambda_1^2(\hat{\R}_4)},
\sqrt{\lambda_i^2(\hat{\R}_4)+\lambda_1^2(\hat{\R}_3)}\},\\
\lambda_i(\hat{\R}_2)&\geq \max\{\sqrt{\lambda_i^2(\hat{\R}_3)+\lambda_1^2(\hat{\R}_5)},
\sqrt{\lambda_i^2(\hat{\R}_5)+\lambda_1^2(\hat{\R}_3)}\}.
\label{e:adivlb2}
\end{align}
\end{theorem}

\begin{proof}
By the Woodbury matrix identity, we have
\begin{align*}
(\G_1+\G_2)^{-1}=\G_1^{-1}-\G_1^{-1}(\G_1^{-1}+\G_2^{-1})^{-1}\G_1^{-1},\\
(\G_1+\G_2)^{-1}=\G_2^{-1}-\G_2^{-1}(\G_1^{-1}+\G_2^{-1})^{-1}\G_2^{-1}.
\end{align*}
Then
\begin{align}
\label{e:inverse}
\G_1^{-1}=(\G_1+\G_2)^{-1}+\G_1^{-1}(\G_1^{-1}+\G_2^{-1})^{-1}\G_1^{-1},\\
\G_2^{-1}=(\G_1+\G_2)^{-1}+\G_2^{-1}(\G_1^{-1}+\G_2^{-1})^{-1}\G_2^{-1}. \nonumber
\end{align}
Since both $\G_1$ and  $\G_2$ are SPD matrices, so are
$\G_1^{-1}, \G_2^{-1}, (\G_1+\G_2)^{-1}$, $\G_1^{-1}(\G_1^{-1}+\G_2^{-1})^{-1}\G_1^{-1}$
and $\G_2^{-1}(\G_1^{-1}+\G_2^{-1})^{-1}\G_2^{-1}$.
Hence, \eqref{e:adivlb1} and \eqref{e:adivlb2} follow from \eqref{e:adlb}.
\end{proof}


We would like to point out that the equalities in \eqref{e:adivlb1}
and \eqref{e:adivlb2} are achievable. For more details, see the following example.
\begin{example}
Let $\G_1=\diag(\alpha_1,\alpha_2,\cdots, \alpha_n)$ for $\alpha_1,\alpha_2,\cdots, \alpha_n$
satisfying \eqref{e:alpha1n}, $\G_2=\beta\I-\G_1$ for some $\beta>\alpha_n$.
In the following, we show that both the equalities in  \eqref{e:adivlb1} and \eqref{e:adivlb2}
are achievable. By some direct calculations, we have
\begin{align*}
\hat{\R}_1&=\diag\left(\frac{1}{\sqrt{\alpha_1}},\frac{1}{\sqrt{\alpha_2}},\cdots, \frac{1}{\sqrt{\alpha_n}}\right), \,\;\hat{\R}_3=\frac{1}{\sqrt{\beta}}\I,\\
\hat{\R}_2&=\diag\left(\frac{1}{\sqrt{\beta-\alpha_1}},\frac{1}{\sqrt{\beta-\alpha_2}},\cdots, \frac{1}{\sqrt{\beta-\alpha_n}}\right),\\
\hat{\R}_4&=\diag\left(\sqrt{\frac{1}{\alpha_1}-\frac{1}{\beta}},
\sqrt{\frac{1}{\alpha_2}-\frac{1}{\beta}},\cdots,
\sqrt{\frac{1}{\alpha_n}-\frac{1}{\beta}}\right),\\
\hat{\R}_5&=\frac{1}{\sqrt{\beta}}\diag\left(\frac{\sqrt{\alpha_1}}{\sqrt{\beta-\alpha_1}},
\frac{\sqrt{\alpha_2}}{\sqrt{\beta-\alpha_2}},\cdots, \frac{\sqrt{\alpha_n}}{\sqrt{\beta-\alpha_n}}\right).
\end{align*}
Then, by \eqref{e:alpha1n}, for $i=1,2,\ldots, n$, we have
\begin{align*}
\lambda_i(\hat{\R}_1)&=\frac{1}{\sqrt{\alpha_{n-i+1}}}, \,\;
\lambda_i(\hat{\R}_2)=\frac{1}{\sqrt{\beta-\alpha_{i}}}, \,\;
\lambda_i(\hat{\R}_3)=\frac{1}{\sqrt{\beta}}, \\
\lambda_i(\hat{\R}_4)&=\sqrt{\frac{1}{\alpha_{n-i+1}}-\frac{1}{\beta}},\,\;
\lambda_i(\hat{\R}_5)=\sqrt{\frac{1}{\beta-\alpha_{i}}-\frac{1}{\beta}}.
\end{align*}
By some simple calculations, one can easily show that
both the equalities in  \eqref{e:adivlb1} and \eqref{e:adivlb2} are reached.
\end{example}

\begin{remark}
It is worth pointing out that \eqref{e:adlb}, \eqref{e:adivlb1} and \eqref{e:adivlb2} cannot be
generalized to $i>1$, i.e., none of the following inequalities
\begin{align*}
\lambda_i(\R_3)\geq \sqrt{\lambda_i^2(\R_1)+\lambda_i^2(\R_2)},\\
\lambda_i(\hat{\R}_1)\geq \sqrt{\lambda_i^2(\hat{\R}_3)+\lambda_i^2(\hat{\R}_4)}, \\
\lambda_i(\hat{\R}_2)\geq \sqrt{\lambda_i^2(\hat{\R}_3)+\lambda_i^2(\hat{\R}_5)}.
\end{align*}
always hold. Indeed, the following example shows this.
\begin{example}
\label{e:ex1}
Let
\begin{align*}
\G_1=\bmx3&0\\0&1\emx, \,\;\G_2=\bmx1&0\\0&8\emx.
\end{align*}
Then
\begin{align*}
\R_1=\bmx \sqrt{3}&0\\0&1\emx, \,\;\R_2=\bmx1&0\\0&\sqrt{8}\emx, \,\;\R_3=\bmx2&0\\0&3\emx
\end{align*}
which implies that
\[
\lambda_2(\R_1)=\sqrt{3}, \,\;\lambda_2(\R_2)=\sqrt{8}, \,\; \lambda_2(\R_3)=3.
\]
Then, one can see that
\[
\lambda_2(\R_3)< \sqrt{\lambda_2^2(\R_1)+\lambda_2^2(\R_2)}.
\]

Furthermore, by some simple calculations, one can easily show that
\begin{align*}
\lambda_2(\hat{\R}_1)< \sqrt{\lambda_2^2(\hat{\R}_3)+\lambda_2^2(\hat{\R}_4)}, \\
\lambda_2(\hat{\R}_2)< \sqrt{\lambda_2^2(\hat{\R}_3)+\lambda_2^2(\hat{\R}_5)}.
\end{align*}
\end{example}
\end{remark}

From Theorems \ref{t:adlb} and \ref{t:adivlb},
one immediately obtains the following monotonic property of successive minima.
\begin{corollary}
\label{c:mono}
Let $\G_1, \G_2\in \mathbb{R}^{n\times n}$ be SPD matrices such that $\G_1-\G_2$ is
also an SPD matrix.
Let $\R_1, \R_2$ and $\hat{\R}_1, \hat{\R}_2$ be defined as in Theorems \ref{t:adlb}
and \ref{t:adivlb}, then for $i=1,2,\ldots, n$, we have
\beq
\label{e:increase}
\lambda_i(\R_1)> \lambda_i(\R_2), \,\;\lambda_i(\hat{\R}_1)< \lambda_i(\hat{\R}_2).
\eeq
\end{corollary}


From Corollary \ref{c:mono}, we have the following result which shows that the
monotonic property of successive minima keeps unchanged by adding a symmetric positive semidefinite
matrix and/or left multiplying a full column rank matrix followed by right multiplying the transpose
of this matrix.
\begin{corollary}
\label{c:mono2}
Let $\G_1, \G_2\in \mathbb{R}^{n\times n}$ be SPD matrices such that $\G_1-\G_2$ is also SPD.
Let $\G\in \mathbb{R}^{m\times m}$ be an arbitrary symmetric positive semidefinite and
$\B\in \mathbb{R}^{n\times m}$ be an arbitrary full column rank matrix.
Denote
\begin{align*}
\R_1&=\mbox{chol}(\G+\B^\top\G_1\B), \, \R_2=\mbox{chol}(\G+\B^\top\G_2\B),\\
\R_3&=\mbox{chol}(\G+\B^\top\G_1^{-1}\B), \, \R_4=\mbox{chol}(\G+\B^\top\G_2^{-1}\B)
\end{align*}
and
\begin{align*}
\hat{\R}_1&=\mbox{chol}((\G+\B^\top\G_1\B)^{-1}),\\
\hat{\R}_2&=\mbox{chol}((\G+\B^\top\G_2\B)^{-1}),\\
\hat{\R}_3&=\mbox{chol}((\G+\B^\top\G_1^{-1}\B)^{-1}),\\
\hat{\R}_4&=\mbox{chol}((\G+\B^\top\G_2^{-1}\B)^{-1}).
\end{align*}
Then
\begin{align}
\label{e:increase1}
\lambda_i(\R_1)> \lambda_i(\R_2), \,\;\lambda_i(\hat{\R}_1)< \lambda_i(\hat{\R}_2),\\
\lambda_i(\R_3)< \lambda_i(\R_4), \,\;\lambda_i(\hat{\R}_3)> \lambda_i(\hat{\R}_4).
\label{e:increase2}
\end{align}
\end{corollary}

\begin{proof}
Since both $\G_1$ and $\G_2$ are SPD, $\G$ is symmetric positive semidefinite and  $\B$ is a full column
rank matrix, one can see that both $\G+\B^\top\G_1\B$ and $\G+\B^\top\G_2\B$ are also SPD.
Moreover, $\G_1-\G_2$ is SPD, implying that
\[
\G+\B^\top\G_1\B-(\G+\B^\top\G_2\B)=\B^\top(\G_1-\G_2)\B
\]
is also SPD. Thus, \eqref{e:increase1} follows from \eqref{e:increase}.

In the following, we show \eqref{e:increase2}.
By \eqref{e:increase1}, we only need to show that $\G_2^{-1}-\G_1^{-1}$ is an SPD matrix.
Let $\G_3=\G_1-G_2$, then by the assumption, $\G_3$ is an SPD matrix. 
Then, by \eqref{e:inverse}, we have
\begin{align*}
\G_2^{-1}&=(\G_2+\G_3)^{-1}+\G_2^{-1}(\G_2^{-1}+\G_3^{-1})^{-1}\G_2^{-1}\\
&=\G_1^{-1}+\G_2^{-1}(\G_2^{-1}+\G_3^{-1})^{-1}\G_2^{-1},
\end{align*}
which implies that $\G_2^{-1}-\G_1^{-1}$ is an SPD matrix and hence \eqref{e:increase2} holds.
\end{proof}

\subsection{Approximating the successive minima}

In this subsection, we propose a lower and an upper bound on the successive minima.

Let $\R\in \mathbb{R}^{n\times n}$ be the R-factor of the QR factorization of a full column
rank matrix $\A$ or a Cholesky factor of an SPD matrix $\G$, then we have
the following result which gives a lower and an upper bound on the successive minima of $\mathcal{L}(\R)$

\begin{proposition}
\label{p:bound}
For $i=1,2,\ldots, n$, we have
\beq
\label{e:bound}
\min_{1\leq j\leq n}|r_{jj}|\leq\lambda_i(\R)\leq \max_{1\leq j\leq i}\|\R_{1:j,j}\|_2.
\eeq
\end{proposition}

\begin{proof}
The second inequality is well-known and can be seen from the definition of successive minima.

The first inequality follows from \eqref{e:nondecrease} and the fact that
$\lambda_1(\R)\geq \min\limits_{1\leq j\leq n}|r_{jj}|$ \cite{LenLL82}.
For the sake  of readability, in the following, we recall its proof from \cite{LenLL82}.
Let $\x\in \mathbb{Z}^{n}$ such that $\lambda_1(\R)=\|\R\x\|_2$ and suppose that the last nonzero entry of $\x$ is $x_i$,
then clearly $\lambda_1(\R)\geq |r_{ii}x_i|\geq |r_{ii}|$.
Hence $\lambda_1(\R)\geq \min\limits_{1\leq j\leq n}|r_{jj}|$ holds.
\end{proof}

\begin{remark}
Note that both of the inequalities in \eqref{e:bound} are achievable.
For example, if $\R=\alpha\I$ for some $\alpha>0$, then one can easily see that
for $i=1,2,\ldots, n$,
\[
\alpha=\min_{1\leq j\leq n}|r_{jj}|=\lambda_i(\R)=\max_{1\leq j\leq i}\|\R_{1:j,j}\|_2.
\]
\end{remark}

\begin{remark}
\label{r:lambdan}
The first inequality in \eqref{e:bound} can be slightly improved when $i=n$. Specifically, we have
$\lambda_n(\R)\geq \sqrt[n]{|\det(\R)|}$.
In fact, by \cite[Proposition 2]{OrdE17}, we have $\prod_{i=1}^n\lambda_i(\R)\geq |\det(\R)|$.
which combing with \eqref{e:nondecrease} implies the above inequality.
\end{remark}

Note that the bounds given by \eqref{e:bound} will became sharper if we use the information of
$\bbR$ to give the upper bound, where $\bbR$ is a lattice reduced upper triangular matrix of $\R$.
Some of the commonly used lattice reduction strategies to achieve this purpose include
the LLL reduction \cite{LenLL82}, KZ reduction \cite{KZ73} \cite{WenC15} and
Minkowski reduction \cite{Min96}.

By Theorem \ref{t:adlb}, Proposition \ref{p:bound} and Remark \ref{r:lambdan},
one can easily obtain the following result:
\begin{corollary}
\label{c:bound}
Let $\G_1, \G_2\in \mathbb{R}^{n\times n}$ be SPD matrices such that $\G_1-\G_2$ is
also an SPD matrix. Let $\R_1, \R_2, \R_3$ be defined as in Theorem \ref{t:adlb},
then for $i=1,2,\ldots, n-1$, we have
\begin{align*}
\lambda_i(\R_3)\geq \max\{\sqrt{\min_{1\leq j\leq n}(r^{(1)}_{jj})^2+\lambda_1^2(\R_2)},\\ \sqrt{\min_{1\leq j\leq n}(r^{(2)}_{jj})^2+\lambda_1^2(\R_1)}\}
\end{align*}
and
\begin{align*}
\lambda_n(\R_3)\geq \max\{\sqrt{|\det(\R_1)|^{2/n}+\lambda_1^2(\R_2)},\\ \sqrt{|\det(\R_2)|^{2/n}+\lambda_1^2(\R_1)}\},
\end{align*}
where $r^{(1)}_{jj}$ and $r^{(2)}_{jj}$ are the $j$-th diagonal entries of $\R_1$ and $\R_2$,
respectively.
\end{corollary}

\section{Improving IF design for C-RAN}
\label{s:IF}


An algorithm for suboptimally solving an optimization problem about maximizing the achievable symmetric rate
for the IF strategy with parallel channel decoding and parallel decompression for C-RAN
has been proposed in \cite{BakN17}.
In this section, we will use some  properties of successive minima, that were developed in Sec. \ref{s:property},
to improve its efficiency.

C-RAN is a promising framework for 5G wireless communication systems.
An end-to-end IF architecture for C-RAN has recently been proposed in \cite{BakN17}.
Its main idea is to employ an  IF source coding \cite{OrdE17}, which can be either symmetric or asymmetric,
to send the channel observations to the central processor.
Then, IF channel coding \cite{ZhaNEG14} is utilized to decode the channel codewords.
By \cite[Theorem 1]{BakN17}, the achievable symmetric rate of the
IF strategy with parallel channel decoding and parallel decompression is
\begin{align}
\label{e:rate}
\mathcal{R}=\max_{d>0, \X\in \mathbb{Z}^{n\times n}, \det(\X)\neq 0}\min_{1\leq i\leq n}\frac{1}{2}\log^{+}\left(\frac{P}{\|\F(d)\x_i\|_2^2}\right)
\end{align}
subject to
\begin{align}
\min_{\bar{\X}\in \mathbb{Z}^{n\times n},\det(\bar{\X})\neq 0}\max_{1\leq i\leq m}\frac{1}{2}\log^{+}(\|\bar{\F}(d)\bar{\x}_i\|_2^2)&\leq C,
\label{e:consd}
\end{align}
with
 \begin{align}
\label{e:F}
\F(d)&=\mbox{chol}((P^{-1}\I+(\B\H)^\top(\B\B^\top+d\I)^{-1}\B\H)^{-1}),\\
\bar{\F}(d)&=\mbox{chol}(d^{-1}(P\B\H(\H\B)^\top+\B\B^\top)+\I),
\label{e:Fbar}
\end{align}
where $\H\in \mathbb{R}^{m\times n}$ is the channel matrix from $n$ users to the $L$ base stations,
$\B\in \mathbb{R}^{m\times m}$ is block diagonal matrix which has $L$ blocks with each of them being
a linear equalizer, $P$ is a constant about the power constraint on the codeword, $C$ is a capacity
and $\log^+(x) \triangleq \max \left( \log(x),0 \right)$.

By \eqref{e:rate}-\eqref{e:Fbar}, one can see that to find a matrix $\hat{\X}$ which maximizes
$\mathcal{R}$, one needs to find $d$ satisfying \eqref{e:consd}
to explicitly form $\F(d)$ (see \eqref{e:F}).
Suppose that $d$ is found, then finding $\hat{\X}$ is equivalent to solving the following problem:
\begin{align}
\label{e:rate2}
\hat{\X}=\argmin_{\X\in \mathbb{Z}^{n\times n}, \det(\X)\neq 0}\max_{1\leq i\leq n}\|\F(d)\x_i\|_2^2.
\end{align}
By the definition of successive minima and \eqref{e:SIVP}, one can see that \eqref{e:rate2} is actually
a SIVP problem which is suppose to be NP-hard.
Moreover, finding $d$'s satisfying \eqref{e:consd} is also time consuming,
Hence a suboptimal algorithm is proposed to solve \eqref{e:rate} in \cite{BakN17}.

In the following, we briefly recall the suboptimal algorithm in \cite{BakN17}.
It is claimed  that $\lambda_n(\F(d))$ (see \eqref{e:F}) is monotonically increasing in $d$ without proof
(in fact this can be seen from Corollary \ref{c:mono2} and \eqref{e:F}), thus the smallest $d$
satisfies
\beq
\label{e:consd2}
\min_{\bar{\X}\in \mathbb{Z}^{n\times n},\det(\bar{\X})\neq 0}\max_{1\leq i\leq m}\frac{1}{2}\log^{+}(\|\bar{\F}(d)\bar{\x}_i\|_2^2)= C
\eeq
is the desired $d$.
In fact, by the definition of successive minima, \eqref{e:consd} is equivalent to $\lambda_n(\bar{\F}(d))\leq \exp(2C)$.
By Corollary \ref{c:mono2} and \eqref{e:Fbar}, one can see that $\lambda_n(\bar{\F}(d))$ is decreasing with $d$.
Thus, the desired $d$ is the one satisfies \eqref{e:consd2}.
After finding $d$ and explicitly forming  $\F(d)$, the LLL reduction is used to find a suboptimal $\hat{\X}$.
A bisection search method coupled with the LLL reduction on $\bar{\F}(d)$ has been proposed in
\cite{BakN17} to find an approximation solution of \eqref{e:consd2} to get a suboptimal $d$.
The bisection method is initialized by setting
$d_{\min}=0$ and $d_{\max}$ large enough such that \eqref{e:consd} holds
($d_{\max}$ is not explicitly given).

In the following, we improve this suboptimal algorithm.
First, instead of using the LLL reduction, for efficiency, we use the PLLL reduction which was
proposed in \cite{XieCA11} followed by size reduction.
The latter has the same performance as the former in this
application, but it is around $\bigO(n)$ times faster than the former.
Second, instead of setting $d_{\min}=0$ and $d_{\max}$ large enough such that \eqref{e:consd} holds
as in \cite{BakN17}, we use a larger $d_{\min}$ and explicitly giving $d_{\max}$.
More details on this are giving as follows.

By the above analysis, finding the smallest $d$ is equivalent to solving $\lambda_n(\bar{\F}(d))=\exp(2C)$.
Since $\lambda_n(\bar{\F}(d))$ is decreasing with $d$, to use the bisection method to find the
desired $d$,
we need to find $d_{\min}$ and $d_{\max}$ such that $\lambda_n(\bar{\F}(d_{\min}))\geq\exp(2C)$
and $\lambda_n(\bar{\F}(d_{\max}))\leq \exp(2C)$, respectively.
To this end, we denote
\beq
\label{e:GF}
\G=\mbox{chol}(P\B\H(\H\B)^\top+\B\B^\top),\,\;
\hat{\F}=\mbox{chol}(\G).
\eeq
Then, by \eqref{e:Fbar}, \eqref{e:GF} and  Corollary \ref{c:bound}, we have
\begin{align}
\label{e:Flbd}
\lambda_n(\bar{\F}(d))&\geq \sqrt{|\det(d^{-1/2}\hat{\F})|^{2/n}+1}\nonumber\\
&= \sqrt{d^{-1}|\det(\hat{\F})|^{2/n}+1}.
\end{align}
Let $d_{\min}$ be the solution of
\[
\sqrt{d^{-1}|\det(\hat{\F})|^{2/n}+1}=\exp(2C),
\]
then by \eqref{e:Flbd}, $d_{\min}$ satisfies $\lambda_n(\bar{\F}(d_{\min}))\geq\exp(2C)$.

By performing the PLLL reduction (see \cite{XieCA11})  followed by size reduction on $\hat{\F}$ to
find an unimodula matrix $\Z$ (i.e., $\Z\in \mathbb{Z}^{n\times n}$ such that $|\det(\Z)|=1$)
such that $\hat{\F}\Z$ is LLL reduced, then $d^{-1/2}\hat{\F}\Z$
is also LLL reduced for any $d>0$.
Let $d_{\max}$ be the solution of
\[
\sqrt{\max_{1\leq i\leq n}\Z_i^\top(d^{-1}\G+\I)\Z_i}=\exp(2C).
\]
then by the definition of $\lambda_n$, one can see that $\lambda_n(\bar{\F}(d_{\max}))\leq \exp(2C)$.

\section{Conclusion}
\label{s:con}

The successive minima of a lattice is important in both communications and cryptography.
In this paper, we investigated some properties of successive minima and then employed
some of them to improve the efficiency of the suboptimal algorithm for solving an optimization
problem about maximizing the achievable rate of the IF for C-RANs in \cite{BakN17}.

\bibliographystyle{IEEEtran}
\bibliography{ref}

\end{document}